\documentclass[conference]{IEEEtran} 
\IEEEoverridecommandlockouts
\usepackage{amsmath,amsfonts} 
\usepackage{amssymb} 
\usepackage{cases} 
\usepackage{relsize} 

\usepackage{algorithmic} 
\usepackage{algorithm} 

\usepackage{makecell} 
\usepackage{array} 
\usepackage[caption=false,font=normalsize]{subfig} 
\usepackage{stfloats} 
\usepackage{threeparttable} 
\usepackage{graphicx} 
\graphicspath{{./figs/}}

\usepackage{textcomp} 
\usepackage{url} 
\usepackage{verbatim} 
\usepackage{cite} 
\usepackage{enumitem} 

\usepackage{xcolor} 
\usepackage{color} 

\def\BibTeX{{\rm B\kern-.05em{\sc i\kern-.025em b}\kern-.08em T\kern-.1667em\lower.7ex\hbox{E}\kern-.125emX}}

\usepackage[switch]{lineno}
\let\oldequation\equation
\let\oldendequation\endequation
\renewenvironment{equation}
   {\linenomathNonumbers\oldequation}
   {\oldendequation\endlinenomath}
\let\oldalign\align
\let\oldendalign\endalign
\renewenvironment{align}
   {\linenomathNonumbers\oldalign}
   {\oldendalign\endlinenomath}
\let\oldnumcases\numcases
\let\oldendnumcases\endnumcases
\renewenvironment{numcases}[1]
   {\begingroup\linenomathNonumbers\oldnumcases{#1}}
   {\oldendnumcases\endlinenomath\endgroup}

\makeatletter  
\newcommand{\biggg}{\bBigg@{3}}  
\newcommand{\Biggg}{\bBigg@{3.5}}  
\newcommand{\bigggg}{\bBigg@{4}}  
\newcommand{\Bigggg}{\bBigg@{4.5}}  
\makeatother  

\newcommand{\Rx}{\mathbb{R}_{\xv}}

\newcommand{\xv}{{\bf x}}

\newcommand{\xvh}{{\widehat {\bf x}}}

\newcommand{\xvhs}{{\widehat {\bf x}}^{\star}}

\newcommand{\zv}{{\bf z}}
\newcommand{\zvs}{{\bf z}^{\star}}

\newcommand{\Kxh}{\mathsf{K}_{\xvh}}

\newcommand{\Kxs}{\mathsf{K}_{\xvh^{\star}}}

\newcommand{\dv}{{\sf d}}
\newcommand{\ev}{{\sf e}}
\newcommand{\av}{{\sf a}}

\newcommand{\vv}{{\sf v}}
\newcommand{\uv}{{\sf u}}

\newcommand{\Gammam}{\mathsf{\Gamma}}
\newcommand{\dsv}{\mathsf{d}^{\star}}
\newcommand{\Kz}{\mathsf{K}_{\zv}}
\newcommand{\Kzs}{\mathsf{K}_{\zv^{\star}}}

\newcommand{\Dm}{\mathsf{D}}

\newcommand{\Em}{{\sf E}}
\newcommand{\Dsm}{\mathsf{D}^{\star}}
\newcommand{\Am}{\mathsf{A}}
\newcommand{\Bm}{\mathsf{B}}
\newcommand{\Cm}{\mathsf{C}}

\newcommand{\Rm}{\mathsf{K}}

\newcommand{\tr}{\mathsf{tr}}
\newcommand{\rank}{\mathsf{r}}

\renewcommand{\det}{\mathsf{det}}
\newcommand{\diag}{\mathsf{diag}}

\newcommand{\zerov}{{\sf 0}}

\newcommand{\E}[2][]{\mathbb{E}_{#1}\left[#2\right]}

\newcommand{\dd}{{\, \rm{d}}}

\newcommand{\T}{\sf T}

\newcommand{\lambdav}{\hbox{\boldmath$\lambda$}}

\def\mindex#1{\index{#1}}

\def\sq{\hbox{\rlap{$\sqcap$}$\sqcup$}}
\def\qed{\ifmmode\sq\else{\unskip\nobreak\hfil
\penalty50\hskip1em\null\nobreak\hfil\sq
\parfillskip=0pt\finalhyphendemerits=0\endgraf}\fi\medskip}

\long\def\defbox#1{\framebox[.9\hsize][c]{\parbox{.85\hsize}{%
\parindent=0pt
\baselineskip=12pt plus .1pt      
\parskip=6pt plus 1.5pt minus 1pt 
#1}}}

\long\def\beginbox#1\endbox{\subsection*{}%
\hbox{\hspace{.05\hsize}\defbox{\medskip#1\bigskip}}%
\subsection*{}}
\def\endbox{}

\newsavebox{\junk}
\savebox{\junk}[1.6mm]{\hbox{$|\!|\!|$}}

\def\limsup{\mathop{\rm lim\ sup}}

\def\bfmath#1{{\mathchoice{\mbox{\boldmath$#1$}}%
{\mbox{\boldmath$#1$}}%
{\mbox{\boldmath$\scriptstyle#1$}}%
{\mbox{\boldmath$\scriptscriptstyle#1$}}}}

\def\bfmY{\bfmath{Y}}

\def\bfmhhaY{\bfmath{\hhaY}} 
\def\bfmhhaY{\hbox to 0pt{$\widehat{\bfmY}$\hss}\widehat{\phantom{\raise 1.25pt\hbox{$\bfmY$}}}}

\def\til={{\widetilde =}}

\def\FRAC#1#2#3{\genfrac{}{}{}{#1}{#2}{#3}}

\def\ddtp{{\mathchoice{\FRAC{1}{d^{\hbox to 2pt{\rm\tiny +\hss}}}{dt}}%
{\FRAC{1}{d^{\hbox to 2pt{\rm\tiny +\hss}}}{dt}}%
{\FRAC{3}{d^{\hbox to 2pt{\rm\tiny +\hss}}}{dt}}%
{\FRAC{3}{d^{\hbox to 2pt{\rm\tiny +\hss}}}{dt}}}}

\def\average#1,#2,{{1\over #2} \sum_{#1}^{#2}}

\def\eye(#1){{\bf(#1)}\quad}

\usepackage{amsthm}

\newtheoremstyle{mybold}
      {3pt}
      {3pt}
      {\itshape}
      {}
      {\bfseries}
      {.}
      { }
      {}

\newtheoremstyle{myremark}
      {3pt}
      {3pt}
      {\normalfont}
      {}
      {\bfseries}
      {.}
      { }
      {}

\theoremstyle{mybold}
\newtheorem{theorem}{Theorem}
\newtheorem{lemma}{Lemma}
\newtheorem{definition}{Definition}

\theoremstyle{myremark}
\newtheorem{remark}{Remark}

\def\eq#1/{(\ref{eq:#1})}

\newcommand{\beqn}[1]{\notes{#1}%
\begin{eqnarray} \elabel{#1}}

\newcommand{\eeqn}{\end{eqnarray} }

\newcommand{\beq}[1]{\notes{#1}%
\begin{equation}\elabel{#1}}

\newcommand{\eeq}{\end{equation}}

\def\bdes{\begin{description}}
\def\edes{\end{description}}

\newcounter{rmnum}

\newcounter{anum}

{\end{list}}

\def\ass(#1:#2){(#1\ref{#1:#2})}

\def\ritem#1{
\item[{\sf \ass(\current_model:#1)}]
}

\newenvironment{recall-ass}[1]{%
\begin{description}
\def\current_model{#1}}{
\end{description}
}

\newcounter{problem}
\newcounter{save@equation}
\newcounter{save@problem}
\makeatletter
\newenvironment{problem}
{\setcounter{save@equation}{\value{equation}} 
\setcounter{save@problem}{\value{problem}} 
\setcounter{problem}{\value{equation}} 
\let\c@equation\c@problem
\subequations
}
{\endsubequations
\setcounter{save@problem}{\value{equation}}%
\setcounter{equation}{\value{save@equation}}%
}

\begin{document}

\title{Joint Lossy Compression for a Vector Gaussian \\ Source under Individual Distortion Criteria}

\author{Shuao Chen, Junyuan Gao, Yuxuan Shi, Yongpeng Wu, Giuseppe Caire, H. Vincent Poor, and Wenjun Zhang 
\thanks{S. Chen, Y. Wu, and W. Zhang are with the Department of Electronic Engineering, Shanghai Jiao Tong
   University, Shanghai 200240, China (e-mails:
   \{shuao.chen, yongpeng.wu, zhangwenjun\}@sjtu.edu.cn)
   (Corresponding author: Yongpeng Wu).
}
\thanks{
   J. Gao is with the Department of Electrical and Electronic Engineering, The Hong Kong Polytechnic University, Hong Kong SAR, China (e-mail: junyuan.gao@polyu.edu.hk).
}
\thanks{
   Y. Shi is with the Department of Networked Intelligence, Pengcheng Laboratory, Shenzhen 410083, China (e-mail: shiyx01@pcl.ac.cn). 
}
\thanks{G. Caire is with the Communications and Information Theory Group, Technische Universit{\"a}t Berlin, Berlin 10587, Germany (e-mail: caire@tu-berlin.de).
}
\thanks{H. V. Poor is with the Department of Electrical and Computer Engineering, Princeton University, Princeton, NJ 08544, USA (e-mail: poor@princeton.edu).
}
}

\maketitle


\begin{abstract}
This paper investigates the joint compression problem of a vector Gaussian source, where an individual distortion constraint is imposed on each source component. 
It is known that the rate-distortion function (RDF) is lower-bounded by the rate derived from the Hadamard inequality, which becomes exact when the semidefinite condition (SDC) holds. 
However, existing works often overlook the case where the SDC is not satisfied.
Moreover, even when the SDC holds, a quantitative characterization of how correlations enable more efficient compression is lacking.
In this work, we refine the results when the SDC is satisfied and derive new theoretical results when the SDC is not satisfied, thereby establishing theoretical limits for practical source compression with correlations.
Specifically, we examine the properties of optimal source reconstruction and provide upper bounds on its dimension, showing that lower-dimensional reconstructions are essential for efficient compression when the SDC does not hold.
Within a scalable two-type correlation (2TC) covariance framework, we prove that the probability of satisfying the SDC decays exponentially with source length, emphasizing the importance of exploring theoretical limits when the SDC is not met.
Additional, we determine the component-wise correlations that a vector source should possess to achieve the Hadamard compression rate, revealing the trade-off between distortion constraints and correlations.
More importantly, by deriving an explicit RDF with correlations incorporated, we quantitatively characterize the gain in compression efficiency achieved by fully leveraging source correlations.
\end{abstract}




\section{Introduction}

The rapid growth of artificial intelligence (AI), Internet of Things (IoT), and edge computing have led to the generation of vast datasets containing components with heterogeneous significance, in which mission-critical elements coexist with non-essential data~\cite{baccour2022pervasive}. Efficient compression of heterogeneous data streams emerges as a critical enabler to meet stringent latency, transmission, and storage demands in next generation communication systems~\cite{chafii2023twelve}.
From an information-theoretic perspective, the fundamental trade-off in lossy compression lies between the compression rate and the fidelity of the reconstructed object~\cite{shannon1959coding}, and is mathematically characterized by the rate-distortion function (RDF)~\cite{berger1971rate}.

The choice of distortion criterion is key to establishing the RDF.
The commonly used distortion criteria include the sum and individual distortion criteria, as follows~\cite{gray1973new,oohama2014indirect}.

\begin{itemize}[leftmargin=10pt]
   \item \textbf{Sum distortion criterion:} It constrains the sum distortion between the vector source and its reconstruction across all \(N\) components, i.e., \(\sum_{i=1}^{N} d_{i} \leq \sum_{i=1}^{N} e_{i}\), where \(d_{i}\) are component-wise distortions and \(e_{i}\) are their corresponding constraints.
   The RDF is obtained by applying the reverse water-filling principle to allocate optimal distortions for the vector Gaussian source~\cite{thomas2006elements}.
   \item \textbf{Individual distortion criteria:}
   The core idea of these criteria is to impose distinct distortion constraints \(d_{i} \leq e_{i}\) on each component.
Widely used in network coding\cite{berger1978multiterminal,gamal1982achievable}, the criteria support varying fidelity needs for components and capture the impact of source correlations on compression rate. Thus, we adopt the criteria in this paper.
\end{itemize}

The current best known results on the RDF under the individual distortion criteria rely on Hadamard's inequality~\cite{horn2012matrix}, which provides a lower bound on the RDF, i.e., \(\Rx([N], \ev) \geq \mathbb{R}^{l}_{\xv}([N], \ev)\)
with the lower rate~\cite{xiao2005compression,lapidoth2010sending,tinguely2008transmitting}
\begin{align} \label{eq:hadamard bound value}
\mathbb{R}^{l}_{\xv}([N], \ev) = \frac{1}{2} \log \frac{\det(\Rm)}{\det(\Em)}.
\end{align}
In Eq.~\eqref{eq:hadamard bound value}, \(\Rm\) denotes the covariance matrix of the vector source, and \(\Em=\diag(e_{1},\cdots,e_{N})\) is a matrix with \(e_i\) representing the distortion constraints for the \(i\)-th component.
The lower rate gives the RDF if and only if the semidefinite condition (SDC)
\begin{align} \label{eq:SDC}
   \Rm \succeq \Em
\end{align}
is satisfied. Existing works often overlook the case where the SDC is not met, leaving it unexplored. Moreover, even when the SDC is satisfied, a quantitative characterization of how correlations enable more efficient compression is lacking.

In this paper, we investigate the joint compression of a vector Gaussian source under the individual distortion criteria. We refine existing results for the case where the SDC in Eq.~\eqref{eq:SDC} is satisfied and obtain new theoretical results for the case where the SDC is not satisfied.
Specifically, in Theorem~\ref{thm:det(Rm-Dsm)}, we provide the properties of optimal source reconstruction under arbitrary covariance and distortion constraints.
We then establish upper bounds on the dimension of the optimal reconstruction in Theorem~\ref{thm:inertia relation}, and find that efficient compression necessitates a lower-dimensional reconstruction when the SDC does not hold.  
Next, we examine the SDC and the corresponding RDF within a scalable two-type correlation (2TC) covariance. In Theorem~\ref{thm:RE_psd_conditions}, we quantitatively establish the theoretical limits for improving compression efficiency by fully exploiting source correlations.
Additionally, we equivalently transform the SDC into a set of scalar inequalities, enabling its analysis from the perspectives of distortion constraints and source correlations. From the distortion-constraint perspective, in Theorem~\ref{thm:sdc exponential asymp}, we prove that the probability of satisfying the SDC decays exponentially with source length, calling for more analytical results for the case where the SDC is not met.
From the source-correlation perspective, in Theorem~\ref{thm:max rho0}, we determine the quantitative component-wise correlations that a vector source should have to achieve the compression rate in Eq.~\eqref{eq:hadamard bound value}.
In summary, we establish new theoretical limits under individual distortion constraints, providing both quantitative results and valuable insights on how correlations enable more efficient compression.

\subsubsection*{Notation}

Uppercase and lowercase boldface letters denote random matrices and vectors. We use sans-serif letters to denote deterministic matrices and vectors.
\(\rank(\av)\) denotes the number of non-zero entries in a vector \(\av\). We use \(\tr(\Am)\), \(\det(\Am)\), \(\rank(\Am)\), \(\diag(\Am)\), and \(\Am^{\T}\) to denote the trace, determinant, rank, diagonal elements, and transposition of a matrix \(\Am\), respectively.
We denote \([x]^+ = \max\{x, 0\}\). 
For an integer \(k > 0\), we denote \([k]=\{1, 2, \cdots, k\}\).
For a real symmetric matrix \(\mathsf{A}\), let \(\mathsf{A}[\mathcal{I}]\) be the principal submatrix indexed by \(\mathcal{I} \subseteq [N]\).
We use the triplet \((n_+,n_-,n_0)\) to denote the counts of positive, negative, and zero eigenvalues of \( \mathsf{A} \), including multiplicities.
For functions \(f(x)\) and \(g(x)\), \( f(x) = O(g(x)) \) means \( \limsup_{x \to \infty} \left|f(x)/g(x) \right| < \infty \), and \( f(x) = \Theta(g(x)) \) means \( \limsup_{x \to \infty} \left|f(x)/g(x) \right| = c \) where \( 0 < c < \infty \).

\section{System Model}
\label{sec:System Model}


\subsection{Setup} \label{subsec:Setup}
Consider an \(N\)-length vector Gaussian source \(\xv \sim \mathcal{N}(\mathsf{0}, \Rm)\) defined on the alphabet \(\mathcal{X}^N\), where each component of the source \(x_i\) is variance-normalized for \(i \in [N]\).
The eigenvalue decomposition (EVD) of the covariance matrix \(\Rm \succ \zerov\) is given by \(\Rm = \mathsf{U} \mathsf{\Lambda} \mathsf{U}^{\T}\), where \(\mathsf{\Lambda} = \diag(\lambda_1, \lambda_2, \cdots, \lambda_N)\) and \(\mathsf{U}\) consists of the corresponding eigenvectors.\footnote{Any positive semidefinite matrix with zero eigenvalues, termed degenerate, can be projected onto a positive definite matrix subspace~\cite{anderson2003introduction}.}
For the joint lossy compression of the vector Gaussian source \(\xv\), the distortion is quantified by the quadratic distortion matrix
\begin{align} \label{eq:distortion matrix definition}
   \Dm = \E{(\xv - \widehat{\xv})(\xv - \widehat{\xv})^{\T}},
\end{align}
where \(\widehat{\xv}\) is the reconstruction of \(\xv\).
For each component \(x_{i}\), the distortion \(d_i\) (i.e., the \(i\)-th diagonal element of \(\Dm\)) satisfies \(d_i \leq e_i\) with the distortion constraint \(e_i \in (0,1]\).
The vector \(\ev = [e_1, e_2, \cdots, e_N]^{\T}\) represents the ordered sequence of individual distortion constraints.
\begin{remark}
   Although we assume each component variance is normalized, it is sufficiently general.
Let \(\widetilde{\xv} \sim \mathcal{N}(\mathsf{0}, \mathsf{\Sigma}_{\widetilde{\xv}})\) be any non-degenerate \(N\)-length source with covariance \(\mathsf{\Sigma}_{\widetilde{\xv}} = \mathsf{V}^{\frac{1}{2}} \Rm \mathsf{V}^{\frac{1}{2}}\), where \(\mathsf{V} = \diag(\sigma_1^2, \sigma_2^2, \cdots, \sigma_N^2)\) is the variance matrix composed of individual variances \(\sigma_i^2 = \mathbb{E}[\widetilde{x}_i^2]\).
For any \(\widetilde{d}_i > 0\), \(i \in [N]\), imposing a quadratic distortion constraint \(\mathbb{E}[(\widetilde{x}_i - \hat{\widetilde{x}}_i)^2] \le \widetilde{d}_i\) on \(\widetilde{x}_i \sim \mathcal{N}(0,\sigma_i^2)\) is equivalent to \(\mathbb{E}[(x_i - \hat{x}_i)^2] \le e_i\), where \(e_i = \widetilde{d}_i/\sigma_i^2\), \(x_i = \sqrt{e_i}\, \widetilde{x}_i/\sqrt{\widetilde{d}_i}\), and \(\hat{x}_i = \sqrt{e_i}\, \hat{\widetilde{x}_i}/\sqrt{\widetilde{d}_i}\). We then have \(x_i \sim \mathcal{N}(0,1)\) for any \(i \in [N]\). 
\hfill\ensuremath{\lozenge}
\end{remark}


\subsection{Definition of RDF}
\label{subsec:Definition of RDF}
Let \(\{\xv(t)\}_{t=1}^\infty\) be an independent and identically distributed (i.i.d.) vector-valued Gaussian source with each \(\xv(t) \sim \mathcal{N}(\zerov, \Rm)\) for all \(t \geq 1\). 
We denote \(k\) i.i.d. realizations of the \(i\)-th source component as \(\xv_i^k \triangleq [x_i(1), \cdots, x_i(k)]\), and define the concatenated source and reconstruction vectors across all \(N\) source components as \(\mathbf{x}^k \triangleq [\mathbf{x}_1^k, \cdots, \mathbf{x}_N^k]^{\T} \in \mathcal{X}^{Nk}\) and \(\widehat{\mathbf{x}}^k \triangleq [\widehat{\mathbf{x}}_1^k, \cdots, \widehat{\mathbf{x}}_N^k]^{\T} \in \widehat{\mathcal{X}}^{Nk}\), respectively.
An \((N, k, 2^{kR})\) lossy code consists of a common encoder \(\mathsf{f}: \mathcal{X}^{Nk} \mapsto \{1,2, \cdots, \lfloor 2^{kR} \rfloor\}\) and a decoder \(\mathsf{g}: \{1,2, \cdots, \lfloor 2^{kR} \rfloor\} \mapsto \widehat{\mathcal{X}}^{Nk}\).
The blockwise average distortion is defined as
\(d (\xv_{i}^{k}, \widehat{\xv}_{i}^{k}) \triangleq \frac{1}{k} \sum_{t=1}^k (x_{i}(t) - \widehat{x}_i(t))^{2}\).
We then define the RDF under the individual distortion criteria.
\begin{definition} \label{def:vector_rdf}
   The pair \((R, \ev)\) is said to be achievable if there exists a sequence of \((N, k, 2^{kR})\) codes such that for all \(\varepsilon_i > 0\) with \(i \in [N]\) and all sufficiently large \(k\),
\begin{align} \label{eq:vector criterion}
\E{d(\xv_{i}^{k}, \widehat{\xv}_{i}^{k})} \leq e_{i} + \varepsilon_{i},
\end{align}
holds. The rate-distortion function under the individual distortion criteria is defined as
\begin{align} \label{eq:vector_rdf_definition}
   \Rx([N], \ev) \triangleq \inf\{ R : (R, \ev) \text{ is achievable} \}.
\end{align}
\end{definition}
For the sum distortion criterion, the RDF can be obtained by replacing the distortion fidelity in Eq.~\eqref{eq:vector criterion} with the sum distortion.
The following lemma provides an optimization formulation of RDF under the individual distortion criteria.
\begin{lemma}[Formulation of RDF under the individual distortion criteria,~\cite{xiao2005compression}] \label{lemma:maxdet}
RDF in Eq.~\eqref{eq:vector_rdf_definition} is given by the solution to
\begin{problem} \label{problem:maxdet}
\begin{align}
\min_{\Dm} \quad & \frac{1}{2} \log \det (\Dm^{-1} \Rm), \label{eq:maxdet obj} \\
\mathrm{s.t.} \quad & \dv \le \ev, \label{eq:maxdet c1} \\
& \zerov \prec \Dm \preceq \Rm, \label{eq:maxdet c2}
\end{align}
\end{problem}
where \(\dv = \diag(\Dm)\) is the vector of diagonal elements of \(\Dm\), and \(\ev = [e_1, \cdots, e_N]^{\T}\) is the normalized distortion constraint vector. \(\dv \le \ev\) means \(d_i \le e_i\) for all \(i \in [N]\).
\end{lemma}
\stepcounter{equation}

The matrix constraint in Eq.~\eqref{eq:maxdet c2} ensures the achievability of the RDF.
We denote the vector \(\dsv = \diag(\Dsm)\), where \(\Dsm\) is the optimal solution to the Max-Det problem in Eq.~\eqref{problem:maxdet}.

\setcounter{theorem}{0}

\section{Source Compression with Arbitrary Covariance}
\label{sec:Source Compression with Arbitrary Covariance}


For the vector Gaussian source compression problem under the individual distortion criteria, the source \(\xv\) can be modeled through a backward test channel as \(\xv = \xvh + \zv\)~\cite{xiao2005compression,nayak2010successive}, where the reconstruction \(\xvh \sim \mathcal{N}(\mathsf{0}, \Kxh)\) and the noise \(\zv \sim \mathcal{N}(\mathsf{0}, \Kz)\). Since \(\xvh\) and \(\zv\) are independent, their covariance matrices satisfy \(\Rm = \Kxh + \Kz\).
When the RDF is achieved, we denote \(\xvh = \xvhs\) and \(\zv = \zvs\), with the positive definite covariance matrix \(\Kzs = \Dsm\). Therefore, the covariance of the optimal reconstruction \(\xvhs\) is \(\Kxs = \Rm - \Dsm\). The following theorem connects the property of \(\xvhs\) to the SDC in Eq.~\eqref{eq:SDC}.
\begin{theorem} \label{thm:det(Rm-Dsm)}
   For the covariance matrix \(\Kxs\) of the optimal reconstruction \(\xvhs\) achieving the RDF, we have
\begin{align} \label{eq:det(Rm-Dsm)}
   \det(\Kxs)
   \begin{cases}
   > 0, & \text{if SDC is satisfied and inactive}, \\
   = 0, & \text{otherwise}.
   \end{cases}
\end{align}
\end{theorem}
\begin{IEEEproof}
We verify the Slater condition to ensure strong duality of the convex problem in Lemma~\ref{lemma:maxdet}. For a bounded and non-trivial \(\Rx([N], \ev)\), the constraint in Eq.~\eqref{eq:maxdet c1} is affine, and there exists a feasible \(\Dm\) in the relative interior of Eq.~\eqref{eq:maxdet c2}. Thus, strict feasibility is satisfied, and the Karush-Kuhn-Tucker (KKT) conditions are both necessary and sufficient for global optimality~\cite{boyd2004convex}. The Lagrangian that omits the constant \(\log (\det(\Rm))/2 \) for the convex Max-Det problem is given by \(\mathcal{L}(\Dm; \mathsf{P}, \mathsf{Q}) = -\frac{1}{2} \log \det(\Dm) - \frac{1}{2} \tr(\mathsf{P} (\Rm - \Dm))- \frac{1}{2} \tr(\mathsf{Q} (\Em - \Dm))\),
where \(\mathsf{P}, \mathsf{Q} \succeq \zerov\), and \(\mathsf{Q}\) is diagonal.
From the stationarity condition, the optimal distortion matrix satisfies
\begin{align} \label{eq:optimal D}
   \Dsm = (\mathsf{P} + \mathsf{Q})^{-1}.
\end{align}
Since \(\Dsm\) is symmetric and \(\mathsf{Q}\) is diagonal, \(\mathsf{P}\) is symmetric. 

\textbf{Case I:} If the SDC is satisfied and inactive, i.e., \(\Rm - \Dsm \succ \zerov\), the complementary slackness condition
\begin{align} \label{eq:slackness1}
   \mathsf{P} (\Rm - \Dsm) = \zerov
\end{align}
implies \(\mathsf{P} = \zerov\). Consequently, since \(\mathsf{Q} \succ \zerov\) is diagonal, it follows that \(\Dsm = \mathsf{Q}^{-1} \succ \zerov\) is diagonal, the other complementary slackness condition
\begin{align} \label{eq:slackness2}
   \mathsf{Q} (\ev - \dsv) = \zerov
\end{align}
leads to \(\dsv = \ev\) and \(\Dsm = \Em\), thereby completing the first part.

\textbf{Case II:} On the boundary of \(\Rm - \Em \succeq \zerov\), i.e., the SDC is satisfied and active, \(\Rm - \Dsm\) is singular.
In this case, \(\Dsm = \Em\) still satisfies the KKT conditions. 
If the SDC is not satisfied, i.e., \(\Rm - \Em \not\succeq \zerov\), Eq.~\eqref{eq:slackness1} no longer implies \(\mathsf{P} = \zerov\). Assuming that \(\mathsf{P} = \zerov\) still holds, Eq.~\eqref{eq:optimal D} implies that \(\Dsm\) is diagonal, and Eq.~\eqref{eq:slackness2} directly gives \(\Dsm = \Em\). This would revert to Case I, contradicting the primal feasibility constraints. We thus have \(\Dsm \neq \Em\) and \(\mathsf{P} \neq \zerov\). From Eq.~\eqref{eq:slackness1}, we obtain the equivalent condition \(\det(\Rm - \Dsm) = 0\) for the second part.
\end{IEEEproof}

In Theorem~\ref{thm:det(Rm-Dsm)}, we establish that the rate in Eq.~\eqref{eq:hadamard bound value} yields an exact RDF, i.e. the SDC is satisfied, if and only if the optimal source reconstruction is non-degenerate.
Conversely, any non-degenerate reconstruction over the alphabet \(\widehat{\mathcal{X}}^{N}\) cannot achieve the Hadamard compression rate. In other words, fewer than \(N\) independent components are sufficient to reconstruct a source with \(N\) independent components. 

Theorem~\ref{thm:det(Rm-Dsm)} examines whether \(\Kxs\) is full rank. The following theorem further gives upper bounds on the rank of \(\Kxs\), which is the number of non-trivial independent components in the optimal reconstruction.

\begin{theorem} \label{thm:inertia relation}
   The dimension of the optimal reconstruction, i.e., the rank of its covariance matrix, is upper-bounded as
\begin{align} \label{eq:inertia relation}
   \rank(\Kxs) &\leq \min\{ N - \rank(\ev - \dsv), n_+(\Rm - \Em)\}.
\end{align}
\end{theorem}
\begin{IEEEproof}
Recalling that \(\Kxs \succeq \zerov\), only the \(\rank(\Kxs)\) independent components in the optimal reconstruction \(\widehat{\xv}^{\star}\) are necessary to achieve the desired distortion fidelity. Relative to the original source \(\xv\), there exists a reduced Gaussian source \(\bar{\xv} \sim \mathcal{N}(\mathsf{0}, \bar{\Rm})\) with length \(\rank(\bar{\mathsf{K}}_{\xvh^{\star}})\) and a corresponding vector distortion constraint \(\bar{\ev}\). Let its optimal reconstruction be \(\bar{\widehat{\xv}}^{\star} \sim \mathcal{N}(\mathsf{0}, \bar{\mathsf{K}}_{\xvh^{\star}})\), which is non-degenerate because \(\rank(\bar{\mathsf{K}}_{\xvh^{\star}}) \leq \rank(\Kxs)\). Each component of \(\bar{\xv}\) satisfies the optimal test channel \(\bar{x}_i = \bar{\hat{x}}_i^{\star} + \bar{z}_i^{\star}\) for \(1 \leq i \leq \rank(\Kxs)\). The non-degeneracy of \(\bar{\widehat{\xv}}^{\star}\) leads to \(\bar{\Rm} - \bar{\Dm}^{\star} \succ \zerov\). From the argument in Case I of the proof of Theorem~\ref{thm:det(Rm-Dsm)}, it follows that
\begin{align} \label{eq:inertia relation1}
\bar{\Dm}^{\star} = \bar{\Em},
\end{align}
and thus we obtain \(\bar{\Rm} - \bar{\Em} \succ \zerov\). On the other hand, for the matrix \(\Rm - \Em\), since a submatrix of dimension \(\rank(\bar{\mathsf{K}}_{\xvh^{\star}})\) is strictly positive definite, it follows that
\begin{align} \label{eq:inertia relation2}
   \rank(\bar{\mathsf{K}}_{\xvh^{\star}}) \leq n_+(\Rm - \Em).
\end{align}
Cauchy's interlace theorem states~\cite{hwang2004cauchy} that for a symmetric matrix \(\Am \in \mathbb{R}^{N \times N}\) with eigenvalues \(\lambdav_i(\Am)\) (in non-decreasing order) and its principal submatrix \(\Bm \in \mathbb{R}^{n \times n}\) (\(n < N\)) with eigenvalues \(\lambdav_i(\Bm)\), we have \(\lambdav_i(\Am) \leq \lambdav_i(\Bm) \leq \lambdav_{i+N-n}(\Am)\).  
Due to \(\lambdav_1(\bar{\mathsf{K}}_{\xvh^{\star}}) > 0\) and setting \(i = 1\), we conclude \(\lambdav_{1+N-n}(\Rm - \Em) > 0\), where \(n = \rank(\bar{\mathsf{K}}_{\xvh^{\star}})\). This implies that \(\Rm - \Em\) has at least \(\rank(\bar{\mathsf{K}}_{\xvh^{\star}})\) positive eigenvalues. Thus, Eq.~\eqref{eq:inertia relation2} is established.
Since \(\bar{\widehat{\xv}}^{\star}\) necessarily exists, \(\bar{\mathsf{K}}_{\xvh^{\star}}\) is full rank, and Eq.~\eqref{eq:inertia relation1} holds, all distortion constraints for the \(\rank(\bar{\mathsf{K}}_{\xvh^{\star}})\) components in the compression are strict, i.e.,
\begin{align} \label{eq:inertia relation3}
\rank(\bar{\mathsf{K}}_{\xvh^{\star}}) \leq N - \rank(\ev - \dsv).
\end{align}

Finally, by exhaustively selecting \(\bar{\xv}\) from \(\xv\), the combination of the relations in Eq.~\eqref{eq:inertia relation2} and Eq.~\eqref{eq:inertia relation3} yields Eq.~\eqref{eq:inertia relation}.
\end{IEEEproof}

In Theorem~\ref{thm:inertia relation}, the first upper bound of \(\rank(\Kxs)\) is given by \(N-\rank(\ev-\dsv)\), which is the number of components for which the distortion constraints are active in the optimal reconstruction.
Milder distortion constraints (larger \(\ev\)) reduce the number of independent components in the optimal reconstruction, while stricter constraints (smaller \(\ev\)) increase independent components.
More importantly, under the individual distortion criteria, the optimal distortion allocation is more intricate than under the sum distortion criterion. For the sum distortion case, the distortion constraint \(\sum_{i=1}^{N} d_{i} \leq \sum_{i=1}^{N} e_{i}\) can be replaced by equality when the rate is non-zero~\cite[Corollary 8.19]{yeung2008information}. In contrast, for the individual case, not all constraints \(d_i \le e_i\) in Eq.~\eqref{eq:maxdet c1} can generally be strict. However, from our first upper bound, \(\rank(\Kxs) \ge 1\) implies \(\rank(\ev - \dsv) \leq N-1\), and at least one distortion constraint is strict, i.e., \(d_i = e_i, \; \exists \, i \in [N]\).

The second upper bound of \(\rank(\Kxs)\) is given by \(n_+(\Rm-\Em)\), which is the number of physically meaningful independent components in the reconstruction \(\xvh\) if we assume \(\Dm=\Em\).
Compared to \(\xvh\), \(\xvhs\) avoids physically unachievable components by replacing them with trivial components (i.e., zero with probability one). To be specific, the optimal reconstruction \(\xvhs\) has at most \(n_+(\Rm-\Em)\) non-trivial independent components.
Therefore, it is easy to deduce that one way to achieve efficient compression is to find a reconstruction with fewer independent components when the SDC does not hold.

\section{Source Compression with 2TC Covariance}
\label{sec:Source Compression with 2TC Covariance}


\subsection{2TC Framework for Covariance Modeling}
\label{subsec:2TC Framework for Covariance Modeling}

To quantify the component-wise correlations in the joint compression of a vector Gaussian source, we consider a covariance matrix based on a hierarchical framework that distinguishes between two types of correlations (components). We refer to it as the 2TC covariance matrix, i.e.,
\begin{align} \label{eq:correlation definition}
   (\Rm)_{ij} =
   \begin{cases}
      1, & \text{if } i = j, \\
      \rho_1, & \text{if } i = 1 \text{ or } j = 1 \text{ and } i \neq j, \\
      \rho_0, & \text{if } i, j > 1 \text{ and } i \neq j,
   \end{cases}
\end{align}
for \(i, j \in [N]\), where \((\Rm)_{ij}\) denotes the \((i, j)\)-th entry of \(\Rm\). Herein, the \textit{central correlation} \(\rho_1\) denotes the correlation between the first central component and the remaining \(N-1\) peripheral components, while the \textit{peripheral correlation} \(\rho_0\) denotes the correlation between the remaining \(N-1\) components. 
To avoid tedious mathematical discussions, we assume all correlations are non-negative, but our results can be extended to the case with negative correlations~\cite{lapidoth2010sending}.
We now present a lemma related to the covariance matrix \(\Rm\) in Eq.~\eqref{eq:correlation definition}, which plays a crucial role in the subsequent analysis.
\addtolength{\topmargin}{0.02in}
\begin{lemma} \label{lemma:determinant}
   Define the matrix \(\Am \triangleq \Rm - \Gammam\), where \(\Gammam = \diag(\gamma_1, \gamma_2, \cdots, \gamma_N)\) with \(0 < \gamma_i \leq 1\). We denote the determinant of \(\Am\) as \(\Delta(\rho_0, \rho_1, \Gammam)\), which is given by 
\begin{align} \label{eq:Delta determinant}
   \Delta(\rho_0, \rho_1, \Gammam) = & \Bigg( \overline{\gamma}_1 + \sum_{i=2}^{N} \frac{\rho_0 \overline{\gamma}_1 - \rho_1^2}{\overline{\gamma}_i - \rho_0} \Bigg)\prod_{i=2}^{N} (\overline{\gamma}_i - \rho_0),
\end{align}
where \(\overline{\gamma}_i = 1-\gamma_i\) for all \(i \in [N]\). 
\end{lemma}
\begin{IEEEproof}
   We first observe that the matrix \(\Am\) can be decomposed as
\(\Am = \Bm + \rho_0 \mathbf{1}\mathbf{1}^{\T}\), where \(\Bm = \begin{pmatrix}
   1 - \rho_0 - \gamma_1 & \vv^{\T} \\
   \vv & \Cm
\end{pmatrix}\),
\(\vv = (\rho_1 - \rho_0, \cdots, \rho_1 - \rho_0)^{\T}\), \(\Cm = \diag(1 - \rho_0 - \gamma_2, \cdots, 1 - \rho_0 - \gamma_N)\), and \(\mathbf{1}\) is the all-ones vector. The matrix \(\Bm\) is an arrowhead matrix, and its determinant is given by \(\det(\Bm) = \xi \prod_{i=2}^{N} (1 - \rho_0 - \gamma_i)\), 
where \(\xi = 1 - \rho_0 - \gamma_1 - \sum_{i=2}^{N} \frac{(\rho_1 - \rho_0)^2}{1 - \rho_0 - \gamma_i}\). The inverse of \(\Bm\) can be expressed as
\(\Bm^{-1} =
\begin{pmatrix}
      0 & \\
      & \Cm^{-1}
\end{pmatrix} + \frac{1}{\xi} \uv \uv^{\T}\), 
where \(\uv = \big(-1, \frac{\rho_1 - \rho_0}{1 - \rho_0 - \gamma_2}, \cdots, \frac{\rho_1 - \rho_0}{1 - \rho_0 - \gamma_N}\big)^{\T}\). Applying the matrix determinant lemma~\cite{harville1998matrix}, we have
\(\Delta(\rho_{0}, \rho_{1}, \Gammam) = \det(\Bm) \cdot \Big(1 + \rho_0 \sum_{i,j} \left(\Bm^{-1}\right)_{i,j}\Big)\). 
By substituting \(\det(\Bm)\) and \(\Bm^{-1}\) into \(\Delta(\rho_{0}, \rho_{1}, \Gammam)\) and then simplifying, we obtain Eq.~\eqref{eq:Delta determinant}.
\end{IEEEproof}

Using Lemma~\ref{lemma:determinant}, the eigenvalues of \(\Rm\) in Eq.~\eqref{eq:correlation definition} are given by \(\lambda_1 = 1 - \rho_0\) with multiplicity \(N-2\), and \(\lambda_{2,3} = \frac{1}{2}(N-2)\rho_{0}\pm\frac{1}{2}\sqrt{(N-2)^{2}\rho_{0}^{2}+4(N-1)\rho_{1}^{2}}\).
Thus, \(\Rm \succ \zerov\) if and only if \((N-2) \rho_0 + 1 - (N-1)\rho_1^2 > 0\) holds.
\begin{remark} \label{remark:2TC extension}

The 2TC framework generalizes previous covariance models~\cite{wagner2008rate,xu2015lossy}, which assumed isotropic correlation \(\rho_0 = \rho_1\). While a fully general covariance matrix would complicate the analysis, the 2TC structure in Eq.~\eqref{eq:correlation definition} yields rather concise results and insights. Importantly, its hierarchical structure supports recursive decomposition and the iterative extension of simpler models. Starting from the two-type correlation case and a known formula for \(\det(\Rm_{2}-\Gammam)\), we can reach the three-type correlation case and derive \(\det(\Rm_{3}-\Gammam)\) by splitting \(N-1\) peripherals into a new central component and the remaining \(N-2\) peripherals, continuing recursively. Ultimately, each component is assigned to a specific correlation level, with higher-level components sharing uniform correlation with all lower-level components.
\hfill\ensuremath{\lozenge}
\end{remark}


\subsection{Characterization of SDC Region}
\label{subsec:Characterization of SDC Region}

The following theorem provides the SDC's equivalent form with \(\Rm\) in Eq.~\eqref{eq:correlation definition} and arbitrary distortion constraints. It also gives a closed-form RDF, in which both the source correlations and the distortion constraints are explicitly incorporated.

\begin{theorem} \label{thm:RE_psd_conditions}
   When \(\Rm\) is the 2TC covariance matrix, the SDC \(\Rm - \Em \succeq \zerov\) holds if and only if
\begin{align}
   e_3 &\leq 1 - \rho_0, \label{eq:e3_condition} \\
   e_2 &\leq 1-\frac{\chi_{3}\rho_{0}^{2}}{1+\chi_{3}\rho_{0}}, \label{eq:e2_condition} \\
   e_1 &\leq 1-\frac{\chi_{2}\rho_{1}^{2}}{1+\chi_{2}\rho_{0}} \label{eq:e1_condition}
\end{align}
hold, where \(\chi_{j}=\sum_{i=j}^{N} \frac{1}{1 - \rho_0 - e_i} \) and we assume \(e_{2} \geq \cdots \geq e_{N}\). The corresponding rate-distortion function is given by
\begin{align}
   \!\!\!\! \Rx([N], \ev) &= \sum_{i=2}^{N} \frac{1}{2} \log \frac{1-\rho_{0}}{e_i} + \frac{1}{2} \log \frac{1 + (N-1)\frac{\rho_0 - \rho_1^2}{1 - \rho_0}}{e_1}. \label{eq:rd sdc1}
\end{align}
When the correlation is isotropic, i.e., \(\rho=\rho_{0}=\rho_{1}\), we have
\begin{align} \label{eq:rd sdc iso}
   \Rx([N], \ev) = \sum_{i=1}^{N} \frac{1}{2} \log \frac{1-\rho}{e_i} + \frac{1}{2} \log \frac{1 + (N-1)\rho}{1-\rho}.
\end{align}
The average rate per component \(\overline{\mathbb{R}}_{\xv}([N], \ev)=\frac{\Rx([N], \ev)}{N}\) is
\begin{align}  \label{eq:alliso rate per asym} 
\overline{\mathbb{R}}_{\xv}([N], \ev) =  
\frac{1}{2} \log \frac{1-\rho}{e_{g}} + \Theta\left(\frac{\log N}{N}\right),
\end{align}
where the geometric mean distortion \(e_{g} = \big(\prod_{i=1}^N e_i\big)^{1/N}\).
\end{theorem}
\begin{IEEEproof}
   Sylvester's criterion states that the equivalent condition for \(\Rm - \Em \succeq \zerov\) is that all principal minors of \(\Rm - \Em\) are non-negative~\cite{horn2012matrix}.
Define two index sets: \(\mathcal{I}_{\backslash \{1\}} \subset [N]\) for principal minors excluding the first element and \(\mathcal{I}_{\cup \{1\}} \subseteq [N]\) for those including it. We analyze these two cases separately.

First, consider the case excluding the first component. By Lemma~\ref{lemma:determinant}, the determinant of \((\Rm - \Em)[\mathcal{I}_{\backslash \{1\}}]\) is given by
\begin{align}
\Delta(\rho_{0},\ev,\mathcal{I}_{\backslash \{1\}}) ={}& \Big( 1 + \!\!\! \sum_{i \in \mathcal{I}_{\backslash \{1\}}} \!\!\!\frac{\rho_0}{1 - \rho_0 - e_i} \Big) \!\!\! \prod_{i \in \mathcal{I}_{\backslash \{1\}}} \!\!\! (1 - \rho_0 - e_i).
\end{align}
If \(e_2 + \rho_{0} \leq 1\), then it follows that \(e_i + \rho_{0} \leq 1\) for all \(i \neq 1\), and \((\Rm - \Em)[\mathcal{I}_{\backslash \{1\}}]\) is positive semidefinite for any \(\mathcal{I}_{\backslash \{1\}}\). If \(\exists \, j \in [N]\) such that \(e_j + \rho_{0} > 1\), the condition
\begin{align} \label{eq:RE_psd_uniqueness}
\frac{\rho_{0}}{\rho_{0} + e_j - 1} \geq 1 + \sum_{i \in \mathcal{I}_{\backslash \{1,j\}}} \frac{\rho_{0}}{1 - \rho_{0} - e_i}
\end{align}
should hold. To prove the uniqueness of \(j\), suppose there exists another \(k \neq j\) such that both \(j\) and \(k\) satisfy the condition in Eq.~\eqref{eq:RE_psd_uniqueness}. The \(2 \times 2\) principal minor corresponding to the subset \(\{j, k\}\) is given by \(\det((\Rm - \Em)[\{j, k\}]) = (1 - e_j)(1 - e_k) - \rho_0^2\). 
Since \(e_j + \rho_{0} > 1\) and \(e_k + \rho_{0} > 1\), we have \(\det((\Rm - \Em)[\{j, k\}]) < 0\). This contradicts the positive semidefiniteness of \(\Rm - \Em\). Therefore, only \(e_2 > 1-\rho_{0}\) satisfies the condition in Eq.~\eqref{eq:RE_psd_uniqueness}, as \(j\) is unique and \(\{e_i\}_{i=2}^{N}\) are in non-increasing order.
For any \(\mathcal{I}_{\backslash \{1,j\}} \subset [N]\), this condition holds, and rewriting it leads directly to Eq.~\eqref{eq:e3_condition} and Eq.~\eqref{eq:e2_condition}.

Next, consider the case including the first component. The matrix \((\Rm - \Em)[\mathcal{I}_{\cup \{1\}}]\) is positive semidefinite if and only if
\begin{align} \label{eq:RE_psd_exclude}
\Big( 1 - e_1 + \sum_{i \in \mathcal{I}_{\backslash \{1\}}} \frac{\rho_0 - \rho_1^2 - \rho_0 e_1}{1 - \rho_0 - e_i} \Big) (1 - \rho_0 - e_2) \geq 0.
\end{align}
holds based on Eq.~\eqref{eq:e3_condition}. By replacing the set \(\mathcal{I}_{\backslash \{1\}}\) with \([N]\) and performing simplification, Eq.~\eqref{eq:RE_psd_exclude} reduces to Eq.~\eqref{eq:e1_condition}.

Finally, we obtain the RDF in Eq.~\eqref{eq:rd sdc1} using the Hadamard rate in Eq.~\eqref{eq:hadamard bound value} and Lemma~\ref{lemma:determinant} to evaluate \(\det(\Rm)\).
\end{IEEEproof}
In contrast to the classical rate-distortion trade-off~\cite{thomas2006elements}, where the impact of distortion on the rate was revealed, we further characterize the impacts of both distortion constraints and source correlations on the compression rate.
Specifically, by fully leveraging source correlations, Eq.~\eqref{eq:rd sdc1} reveals the theoretical limits for improving compression efficiency. Stronger correlations \(\rho_0\) and \(\rho_1\) between components result in a lower compression rate while maintaining the same distortion requirements.
Moreover, the complicated interactions between distortion constraints \(\{e_i\}_{i=1}^{N}\) and correlations \(\rho_0\) and \(\rho_1\) are captured in Eqs.~\eqref{eq:e3_condition}-\eqref{eq:e1_condition}, which define the SDC region. It can be analyzed from two perspectives: distortion constraints and source correlations.

\subsubsection{A Distortion-Constraint Perspective}
For a given vector source \(\xv \sim \mathcal{N}(\mathsf{0},\Rm)\) with \(\Rm\) in Eq.~\eqref{eq:correlation definition}, we quantify the probability of event \(\mathcal{A}_{0} = \{\mathbf{E} \preceq \Rm\}\) under uniformly distributed distortion constraints \(e_i \sim \mathcal{U}[0,1]\) for all \(i \in [N]\).
\begin{theorem} \label{thm:sdc exponential asymp}
   \(P(\mathcal{A}_0)\) takes the form of an \(N\)-fold integral in Eq.~\eqref{eq:SDC prob},
\begin{figure*}[t]
   \begin{align}
      P(\mathcal{A}_0) ={}& (N-1)\underbrace{\int_{0}^{1-\rho_{0}}  \cdots \int_{0}^{1-\rho_{0}}}_{N-2 \text{ times}}  \int_{\max\{e_{3},\cdots,e_{N}\}}^{1-\rho_{0}-\frac{\rho_{1}^{2}-\rho_{0}}{1-\left(\rho_{1}^{2}-\rho_{0}\right)\chi_{3}}}  \int_{0}^{1-\frac{\chi_{2}\rho_{1}^{2}}{1+\chi_{2}\rho_{0}}} \, \dd e_{1} \dd e_{2} \dd e_{3} \cdots \dd e_{N}, \label{eq:SDC prob}
   \end{align}
   \hrulefill
   \vspace*{-13pt}
\end{figure*}
and its asymptotics with \(N \to \infty\) is given by 
\begin{align} \label{eq:sdc exponential asymp}
   P(\mathcal{A}_0) = e^{N\log(1-\rho_0) + O(\log N)}.
\end{align} 
\end{theorem}
\begin{IEEEproof}
   The \(N\)-fold integral in Eq.~\eqref{eq:SDC prob} is obtained from Eqs.~\eqref{eq:e3_condition}-\eqref{eq:e1_condition} in Theorem~\ref{thm:RE_psd_conditions}, where \(N-1\) arises from the ordering of the constraints \(\{e_{i}\}_{i=2}^{N}\).
   Moreover, by comparing the upper bounds of \(e_2\) in Eq.~\eqref{eq:e1_condition} with Eq.~\eqref{eq:e2_condition}, we have
\begin{align}  
   &\left(1-\rho_{0}-\frac{\rho_{1}^{2}-\rho_{0}}{1-\left(\rho_{1}^{2}-\rho_{0}\right)\chi_{3}}\right)-\left(1-\frac{\chi_{3}\rho_{0}^{2}}{1+\chi_{3}\rho_{0}}\right) \notag \\  
   ={}& \frac{-\rho_{1}^{2}}{\left(1-\left(\rho_{1}^{2}-\rho_{0}\right)\chi_{3}\right)\left(1+\chi_{3}\rho_{0}\right)} \leq 0.
\end{align}  
Thus, the upper bound of \(e_2\) in Eq.~\eqref{eq:SDC prob} is improved to prevent the divergence of the integral.   
   To analyze the asymptotic behavior of this \(N\)-fold integral, we consider the inner double integral over \(e_1\) and \(e_2\). We observe that
\begin{align}
   \int_{\max\{e_{3},\cdots,e_{N}\}}^{1-\rho_{0}-\frac{\rho_{1}^{2}-\rho_{0}}{1-\left(\rho_{1}^{2}-\rho_{0}\right)\chi_{3}}}  \int_{0}^{1-\frac{\chi_{2}\rho_{1}^{2}}{1+\chi_{2}\rho_{0}}} \, \dd e_{1} \dd e_{2} < 1.
\end{align} 
This is because the integration domain of this double integral, \(\mathcal{D}_2\), is a subset of \([0,1]^2\). Thus, we obtain an asymptotic upper bound for \(P(\mathcal{A}_0)\) as given in Eq.~\eqref{eq:sdc exponential asymp}.
\end{IEEEproof}

In Theorem~\ref{thm:sdc exponential asymp}, we provide the exact \(N\)-fold integral form of \(P(\mathcal{A}_0)\) in Eq.~\eqref{eq:SDC prob}.
We then analyze its asymptotic behavior with respect to \(N\). 
When there is no prior information on each source component, it follows that the probability of satisfying the SDC decays exponentially with the source length at a rate of \(-\log(1-\rho_0)\) from Eq.~\eqref{eq:sdc exponential asymp}.
This highlights the limitations of previous studies, which focused solely on the case where the SDC holds. The theoretical results and valuable insights when the SDC is not met in this work are of great importance.

\subsubsection{A Source-Correlation Perspective}
For given distortion constraints, we determine the quantitative component-wise correlations that a vector source should have to achieve the compression rate in Eq.~\eqref{eq:hadamard bound value}.
The correlation region that satisfies the SDC is
\(\mathcal{C}_{0} = \left\{ (\rho_0, \rho_1) \mid 0 \leq \rho_0 \leq \rho_0^m, \, 0 \leq \rho_1 \leq \rho_1^m \right\}\), 
where \(\rho_1^m = \big(\frac{1}{\chi_{2}}+\rho_{0}\big)^{\frac{1}{2}}\!\!\left(1-e_{1}\right)^{\frac{1}{2}}\) due to Eq.~\eqref{eq:e1_condition}
and satisfies \(\rho_1^m = 0\) if \(\rho_0 = \rho_0^m\).
Since polynomials of degree five or higher generally lack analytical solutions, explicitly solving for \(\rho_0^m\) in Eq.~\eqref{eq:e2_condition} (an \((N-1)\)-th degree polynomial) becomes intractable. If the largest peripheral distortion constraint \(e_2\) is not unique, Eq.~\eqref{eq:e3_condition} directly implies \(\rho_0^m = 1 - e_2\).
Otherwise, we provide bounds and approximation for \(\rho_0^m\).

\begin{figure*}[htbp]
\centering
\begin{minipage}[t]{0.32\textwidth}
   \includegraphics[height=4.2cm]{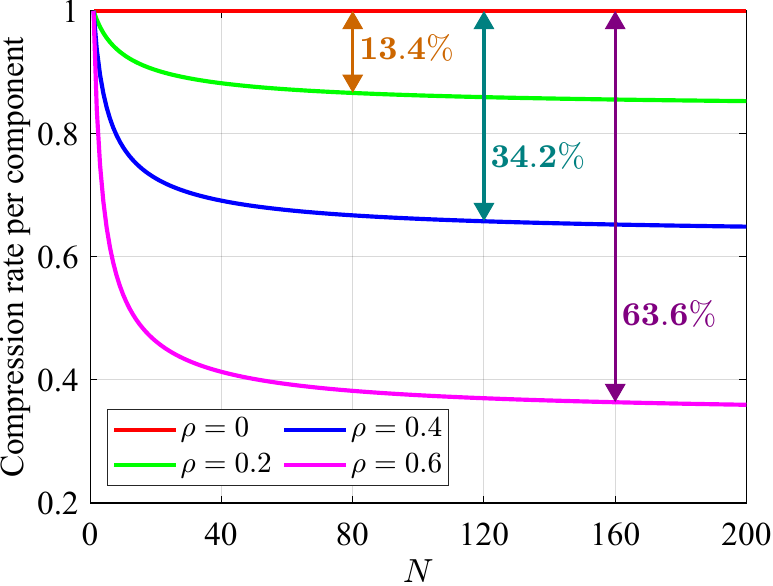}
   \caption{Compression rate per component in Eqs.~\eqref{eq:rd sdc iso} and \eqref{eq:alliso rate per asym} versus source length under identical distortion constraint \(e=0.25\) in Theorems~\ref{thm:RE_psd_conditions}.}
   \label{fig:iso_uniform_2D_N}
\end{minipage}
\hfill
\begin{minipage}[t]{0.32\textwidth}
      \includegraphics[height=4.2cm]{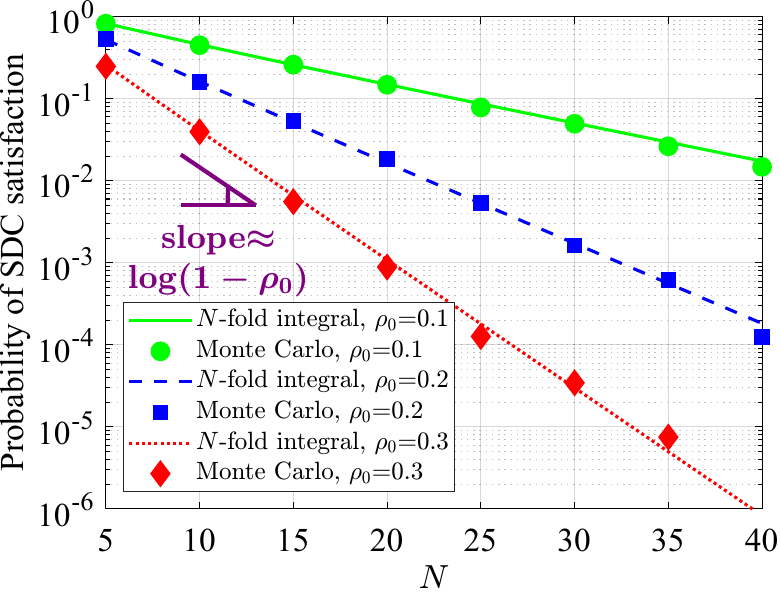}
      \caption{The probability of SDC satisfaction versus the source length in Theorem~\ref{thm:sdc exponential asymp} for \(\rho_1=0.45\) and all distortion constraints \(e_i \sim \mathcal{U}[0,1]\).}
      \label{fig:SDC_prob_N_2}
\end{minipage}
\hfill
\begin{minipage}[t]{0.32\textwidth}
      \includegraphics[height=4.2cm]{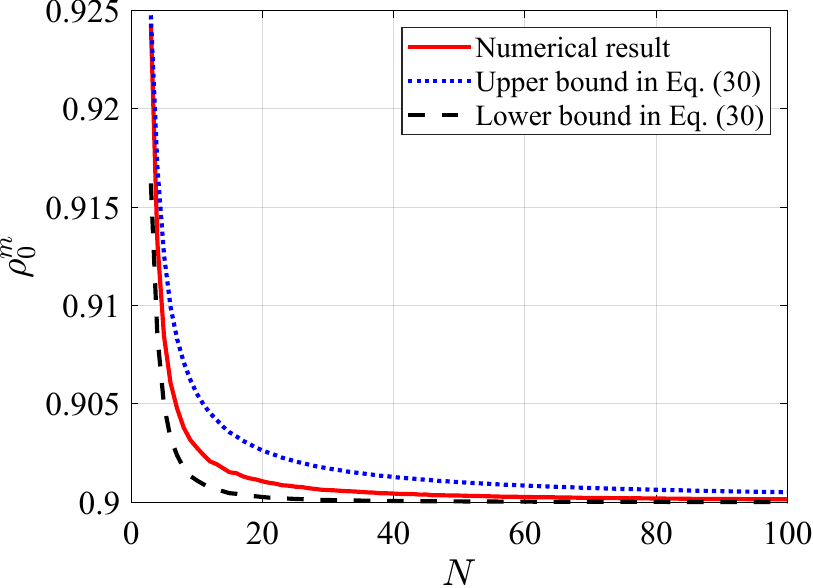}
      \caption{Evaluation of \(\rho_0^m\) versus source length in Theorem~\ref{thm:max rho0} with fixed \(e_2 = 0.1\).}
      \label{fig:rho0m_e2_fixed}
\end{minipage}
\vspace{-1.5em}
\end{figure*}

\begin{theorem} \label{thm:max rho0}
   For a distortion constraint vector \(\ev\), the maximum peripheral correlation \(\rho_0^m\) satisfying the SDC is bounded as
\begin{align} \label{eq:max rho0 bounds}
   1 - e_2 + \frac{c_l}{N} \leq \rho_0^m \leq 1 - e_2 + \frac{c_u}{N-1},
\end{align}
   where \(c_l = 2(e_{2}-e_{3})\Big(\sqrt{\frac{1-e_{3}}{1-e_{2}}}+1\Big)^{-1}\)
   and \(c_u = e_2 - \overline{e}_3\),
   with \(\overline{e}_3 = \frac{1}{N-2}\sum_{i=3}^{N} e_i\).
   For large \(N\), the quantity \(\rho_0^m\) satisfies
   \begin{align} \label{eq:rho0m theta}
      \rho_0^m = 1 - e_2 + \Theta\left(\frac{1}{N}\right).
   \end{align}
\end{theorem}
\begin{IEEEproof}
   We define two functions \(f(\rho_0) = \frac{\rho_0}{\rho_0 + e_2 - 1}\) and \(g(\rho_0) = 1 + \sum_{i=3}^{N} \frac{\rho_0}{1 - \rho_0 - e_i}\). \(f(\rho_0)\) is strictly decreasing, while \(g(\rho_0)\) is piecewise increasing within \(\rho_{0} \in (1 - e_i, 1 - e_{i+1})\) for \(2 \leq i \leq N-1\). We note that \(\rho_0^m\) is the solution to \(f(\rho_0) = g(\rho_0)\) when \(\rho_{0} \in (1 - e_2, 1 - e_3)\).

   To establish a lower bound, we define a strictly increasing function \(g_{1}(\rho_0) = 1 + \frac{(N-2)\rho_0}{1 - \rho_0 - e_3} \geq g(\rho_0)\) in the interval \((0, 1 - e_3)\), with equality if and only if \(e_i = e_3\) for all \(i > 3\).
   Due to the continuity of \(f(\rho_0)\), there exists a lower bound \(\rho_0^l \in (1 - e_2, 1 - e_3)\) for \(\rho_0^m\) such that \(f(\rho_0^l) \geq f(\rho_0^m)\). Solving from \(g_{1}(\rho_0) = f(\rho_0)\), we obtain
   \begin{align} \label{eq:rho0m lower origin}
   \!\!\! \rho_0^l = (1 - e_2) \frac{(N-2) + \sqrt{(N-2)^2 + 4(N-1)\frac{1-e_3}{1-e_2}}}{2(N-1)}.
   \end{align}
   By further derivation of \(\rho_0^l\) in Eq.~\eqref{eq:rho0m lower origin}, a concise lower bound is obtained in Eq.~\eqref{eq:max rho0 bounds} with equality if and only if \(N = 2\).
   
   For the upper bound of \(\rho_0\), we first derive a lower bound for \(g(\rho_0)\) using Jensen's inequality, which is \(g_2(\rho_0) = 1+(N-2)\big(1-\rho_{0}-\frac{1}{N-2}\sum_{i=3}^{N}e_{i}\big)^{-1}\rho_{0}\) 
   for \(\rho_0 \in (0, 1 - e_3)\). We solve \(g_2(\rho_0) = f(\rho_0)\) to determine the upper bound
   \begin{align}
   \!\!\! \rho_0^u = (1 - e_2) \frac{(N-3) + \sqrt{(N-3)^2 + 4(N-2)\frac{1 - \overline{e}_3}{1 - e_2}}}{2(N-2)}.
   \end{align}
   Since \(\sqrt{(N-3)^{2}+4(N-2)\frac{1-\overline{e}_{3}}{1-e_{2}}}\ge N-1\) holds for \(N \geq 2\), we have a more concise upper bound in Eq.~\eqref{eq:max rho0 bounds}.
   The asymptotics of \(\rho_0^m\) in Eq.~\eqref{eq:rho0m theta} is obtained by combining the lower and upper bounds in Eq.~\eqref{eq:max rho0 bounds}.
\end{IEEEproof}

In Theorem~\ref{thm:max rho0}, we derive the lower and upper bounds and the asymptotic approximation for \(\rho_0^m\). By combining \(\rho_1^m\) already obtained, we fill the correlation region \(\mathcal{C}_{0}\) to characterize the SDC.
From Eq.~\eqref{eq:rho0m theta}, we observe that for a vector source with \(\Rm\) in Eq.~\eqref{eq:correlation definition}, if it can be compressed at the Hadamard lower rate in Eq.~\eqref{eq:hadamard bound value} under constraints \(\{e_i\}_{i=1}^N\), the gap between the maximum peripheral correlation \(\rho_{0}^{m}\) and \(1-e_2\) decreases inversely with source length, where \(e_2\) is the mildest peripheral constraint. This implies a trade-off: a source with stronger correlations cannot be compressed at the Hadamard rate under excessively mild distortion constraints.

\subsection{Simulation Results}
\label{subsec:Simulation Results}
In this subsection, we compare our analytical results from Theorems~\ref{thm:RE_psd_conditions}, \ref{thm:sdc exponential asymp}, and \ref{thm:max rho0} with numerical simulations. 

For Theorem~\ref{thm:RE_psd_conditions}, Fig.~\ref{fig:iso_uniform_2D_N} plots the average compression rate per component in Eqs.~\eqref{eq:rd sdc iso} and \eqref{eq:alliso rate per asym} as a function of the number of components, with fixed distortion constraints \(e = e_i = 0.25\) for all \(i \in [N]\). we observe the gains in compression efficiency achieved by fully leveraging correlations. 
Compared to the independent case, the compression cost per component is significantly reduced due to correlations. 
For \(N = 80\) and \(\rho = 0.2\), the compression rate decreases by 13.5\%, for \(N = 120\) and \(\rho = 0.4\), it decreases by 34.2\%, and for \(N = 160\) and \(\rho = 0.6\), it decreases by 63.6\%. 
According to Eq.~\eqref{eq:alliso rate per asym}, as \(N\) increases, \(\overline{\mathbb{R}}_{\xv}([N], \ev)\) converges to \(\frac{1}{2} \log \frac{1-\rho}{e}\). In summary, effectively leveraging correlations in data is a crucial strategy for reducing storage and processing overhead.

For Theorem~\ref{thm:sdc exponential asymp}, Fig.~\ref{fig:SDC_prob_N_2} shows simulation results for \(P(\mathcal{A}_0)\) versus \(N\), based on \(10^{6}\) Monte Carlo trials for each \(N\). \(\{e_i\}_{i=1}^N\) are randomly sampled from \([0,1]\) and the central correlation \(\rho_1 = 0.45\) is fixed. 
The Monte Carlo results align with our \(N\)-fold integral analysis in Eqs.~\eqref{eq:SDC prob} and \eqref{eq:sdc exponential asymp}, showing that the probability of SDC satisfaction decays with \(N\) at a rate of \(\log(1-\rho_0)\). 
This highlights the need for future research to focus more on the case where SDC is not satisfied.

For Theorem~\ref{thm:max rho0}, in Fig.~\ref{fig:rho0m_e2_fixed}, we conduct simulations for \(\rho_0^m\) using results from \(10^{4}\) Monte Carlo trials, with \(\{e_i\}_{i=3}^{N}\) randomly sampled.   
We fix \(e_2 = 0.1\) and observe that \(\rho_0^m\) gradually converges to \(1 - e_2\) as \(N\) increases, while the analytical bounds in Eq.~\eqref{eq:max rho0 bounds} remain remarkably tight to the numerically obtained \(\rho_0^m\).
This verifies that our concise results achieve rather accurate approximations to the analytically intractable root of an \((N-1)\)-th degree equation.
\section{Conclusion}
\label{sec:Conclusion}
In this work, we have investigated the joint compression of a vector Gaussian source under the individual distortion criteria, refined existing results when the SDC is satisfied and obtained new theoretical results when the SDC is not satisfied. 
Our analysis has showed that efficient compression involves a lower-dimensional reconstruction when the SDC does not hold. 
Additionally, we have revealed that the probability of satisfying the SDC decays exponentially with source length, emphasizing the need to consider the case where the SDC is not met. 
Furthermore, we have found that to achieve compression at the Hadamard lower rate under milder distortion constraints, the source should exhibit weaker component-wise correlations.
More importantly, from the derived RDF with explicitly incorporated correlations, we have quantitatively characterized the gains of leveraging correlations between source components to reduce storage and processing costs.

\bibliographystyle{IEEEtran}
\bibliography{references}


\clearpage
\appendices


\end{document}